\documentclass{article}
\usepackage{spconf,amsmath,graphicx}

\usepackage{cite}
\usepackage{amssymb,amsfonts}
\usepackage{textcomp}
\usepackage{xcolor}
\usepackage{algorithm}
\usepackage{algpseudocode}
\usepackage{amsthm}

\usepackage{balance}

\theoremstyle{definition}
\newtheorem{definition}{Definition}
\newtheorem{proposition}{Proposition}

\newtheorem{corollary}{Corollary}

{
	\theoremstyle{plain}
	
}

\newtheorem{innercustomgeneric}{\customgenericname}
\providecommand{\customgenericname}{}
\newcommand{\newcustomtheorem}[2]{%
	\newenvironment{#1}[1]
	{%
		\renewcommand\customgenericname{#2}%
		\renewcommand\theinnercustomgeneric{##1}%
		\innercustomgeneric
	}
	{\endinnercustomgeneric}
}

\DeclareMathOperator*{\argmin}{arg\,min}
\newcustomtheorem{customthm}{Theorem}
\newcustomtheorem{customlemma}{Lemma}
\newcustomtheorem{customDef}{Definition}

\title{Unique Bispectrum Inversion for Signals with \\Finite Spectral/Temporal Support}
\name{Samuel Pinilla$^{\dag}$, Kumar Vijay Mishra$^{\ddag}$ and Brian M. Sadler$^{\ddag}$ \thanks{This research was sponsored by the Army Research Office/Laboratory. K. V. M. acknowledges support from the National Academies of Sciences, Engineering, and Medicine via Army Research Laboratory Harry Diamond Distinguished Postdoctoral Fellowship. S. P. acknowledges support from the EMET Research Institute in Colombia.} \vspace{-1em}}
\address{$^{\dag}$Science and Technology Facilities Council, Turing Hub at Harwell, Oxfordshire OX110FA, UK\\ 
	$^{\ddag}$United States CCDC Army Research Laboratory, Adelphi, MD 20783 USA \vspace{-1em}}
\begin{document}
\setlength{\abovedisplayskip}{3pt}
\setlength{\belowdisplayskip}{3pt}

\ninept
\maketitle
\begin{abstract}
Retrieving a signal from its triple correlation spectrum, also called bispectrum, arises in a wide range of signal processing problems. Conventional methods do not provide an accurate inversion of bispectrum to the underlying signal. In this paper, we present an approach that uniquely recovers signals with finite spectral support (band-limited signals) from at least $3B$ measurements of its bispectrum function (BF), where $B$ is the signal's bandwidth. Our approach also extends to time-limited signals. We propose a two-step trust region algorithm that minimizes a non-convex objective function. First, we approximate the signal by a spectral algorithm and then refine the attained initialization based on a sequence of gradient iterations. Numerical experiments suggest that our proposed algorithm is able to estimate band-/time-limited signals from its BF for both complete and undersampled observations.
\end{abstract}
 \vspace{-1em}
\section{Introduction}\vspace{-0.5em}
In many science and engineering problems, the information conveyed by the autocorrelation or power spectrum of a signal does not sufficiently characterize the signal \cite{mendel1991tutorial}. Extracting information from non-Gaussian and nonlinear signals usually requires computing the higher-order statistics and their Fourier transforms (or polyspectra) \cite{nikias1993signal}. The higher-order spectra analysis may present situations where it is desired to estimate the signal from its polyspectra measurements. For example, in weather radars \cite{mishra2016deployment}, filtering the Gaussian meteorological echoes \cite{bringi2020retrieval} from the non-Gaussian clutter such as wind turbines \cite{hood2010automatic} and anomalous propagation \cite{thurai2017initial} requires inverting the polyspectra samples to the original signal \cite{doviak2014doppler}. In this paper, we focus on the inversion of the third-order cumulant spectrum, also known as the bispectrum \cite{tryon1981bispectrum,nikias1987bispectrum}.

The bispectrum finds applicability in fields such as optics \cite{sadler1992shift}, time-series analysis \cite{godfrey1965exploratory}, seismology \cite{hasselmann1963time}, holography \cite{sato1977bispectral}, and radar \cite{chen2008feature}. 
The \textit{bispectrum function} (BF) is a two-dimensional mapping of circularly-translated versions of a signal $\boldsymbol{x}\in \mathbb{C}^{N}$ defined in terms of its Fourier transform $\boldsymbol{y}\in \mathbb{C}^{N}$ as \cite{matsuoka1984phase}
\begin{align}
	\boldsymbol{B}[k_{1},k_{2}] = \boldsymbol{y}[k_{1}]\overline{\boldsymbol{y}}[k_{2}]\boldsymbol{y}[k_{1}-k_{2}],
	\label{eq:BF}
\end{align}
where $N$ is the size of the signals and $k_{1},k_{2}=0,\dots,N-1$ are Fourier frequencies. This formal model represents the contribution to the cube of the signal from the product of three Fourier components whose frequencies sum to zero \cite{matsuoka1984phase,hasselmann1963time}. 

Classical BF studies \cite{matsuoka1984phase,marron1990unwrapping,rangoussi1990fir} have explored the benefits of this function to determine the wavelet phase of the underlying pulse and when the signal is band-limited (its Fourier transform contains consecutive zeros). For this scenario, the reconstruction procedure first estimates the translations  \cite{brillinger1977identification,lii1982deconvolution,matsuoka1984phase,bartelt1984phase}. Given these estimates, the underlying pulse $\boldsymbol{x}$ is easily retrieved by aligning all observations and then averaging to reduce the noise. However, when the size of the signal increases, the computational complexity of the algorithm becomes intractable. Further, these methods lack performance and uniqueness guarantees to estimate the signal.

More recent works \cite{bendory2017bispectrum} have proposed convex and non-convex optimizations to estimate a particular set of signals with finite spectral support from its BF. This work also suggests that the non-convex method performs better than the convex approach. Contrary to the classical approaches \cite{brillinger1977identification,lii1982deconvolution,matsuoka1984phase}, these recent methods estimate the signal directly, without estimating the shifts, and therefore reducing computational complexity compared with the classical strategies. While \cite{bendory2017bispectrum} provides some theoretical results to analyze the convergence of the non-convex algorithms, their performance is yet to be characterized. Moreover, these methods were developed for only some specific scenarios, namely signals with finite spectral support. However, signals with finite temporal support are also frequently encountered in applications such as radar  \cite{theron1999ultrawide,chen2008mimo,pinilla2022phase,pinilla2021wavemax}. Hence, a more general approach with theoretical uniqueness results and a recovery algorithm is highly desired.

In this paper, we present a uniqueness result to estimate a signal with finite spectral or temporal support from its BF. We show that the underlying signal can be recovered from at least $3B$ ($3S$) measurements where $B$ ($S$) is the bandwidth (time duration). Additionally, we propose a trust region algorithm that minimizes a non-convex objective function to iteratively estimate the band-/time-limited signal-of-interest. We first obtain an initial estimate of the signal through an iterative spectral algorithm. Then, this initialization is refined through a sequence of gradient iterations. We present results on uniqueness and sampling complexity (minimum number of required samples) for recovering band-/time-limited signals from its BF hitherto unreported in the existing literature. Unlike previous studies \cite{bendory2017bispectrum,matsuoka1984phase}, where the gradient descent rule is unable to accurately solve this highly ill-posed problem, we exploit the continuity of the spectrum through the Paley-Wiener theorem to change the initialization method leading to a unique solution. Our numerical experiments show that the proposed algorithm recovers band-/time-limited signals from its BF even when the latter is undersampled.


Throughout this paper, we use boldface lowercase and uppercase letters for vectors and matrices, respectively. The sets are denoted by calligraphic letters and $\lvert \cdot \rvert$ represents the cardinality of the set. The conjugate and conjugate transpose of the vector $\mathbf{w}\in \mathbb{C}^{N}$ are denoted as $\overline{\mathbf{w}}\in \mathbb{C}^{N}$ and $\mathbf{w}^{H}\in \mathbb{C}^{N}$, respectively. The $n$-th entry of a vector $\mathbf{w}$, which is assumed to be periodic, is $\mathbf{w}[n]$. The $(k,l)$-th entry of a matrix $\boldsymbol{A}$ is $\boldsymbol{A}[k,l]$. We denote the Fourier transform of a vector and its conjugate reflected version (that is, $\hat{\mathbf{w}}[n] := \overline{\mathbf{w}}[-n]$) by $\tilde{\mathbf{w}}$ and $\hat{\mathbf{w}}$, respectively. Additionally, we use $\odot$ and $*$ for the Hadamard (point-wise) product, and convolution, respectively; $\sqrt{\cdot}$ is the point-wise square root; superscript within parentheses as $(\cdot)^{(t)}$ indicates the value at $t$-th iteration;  $\mathcal{R}(\cdot)$ denotes the real part of its complex argument; and $\omega=e^{\frac{2\pi i}{n}}$ is the $n$-th root of unity. \vspace{-1em}

\section{Problem Formulation}
\label{sec:problem}\vspace{-0.5em}
The BF defined in \eqref{eq:BF} is a map $\mathbb{C}^{N}\rightarrow \mathbb{C}^{N\times N}$ that has four types of symmetry \cite{matsuoka1984phase}, usually called \textit{trivial ambiguities} as follows:
\begin{enumerate}
	\item $\boldsymbol{B}(k_{1},k_{2}) = \boldsymbol{B}(k_{2},k_{1})$.
	\item $\boldsymbol{B}(k_{1},k_{2}) = \boldsymbol{B}(k_{1},-k_{1} - k_{2})$.
	\item $\boldsymbol{B}(k_{1},k_{2}) = \boldsymbol{B}(k_{2},-k_{1} - k_{2})$.
	\item $\boldsymbol{B}(k_{1},k_{2}) = \boldsymbol{B}^{*}(-k_{1},-k_{2})$ if $\boldsymbol{x}\in \mathbb{R}^{N}$.
\end{enumerate}
We introduce the following useful definitions.
\begin{definition}[$B$-Band-Limitedness]
	The signal $\boldsymbol{x}\in \mathbb{C}^{N}$ is a $B$-band-limited signal if its Fourier transform $\tilde{\boldsymbol{x}}\in \mathbb{C}^{N}$ contains $N-B$ consecutive zeros. That is, there exists $k$ such that $\tilde{\boldsymbol{x}}[k]=\cdots=\tilde{\boldsymbol{x}}[N+k-B-1]=0$.
	\label{def:bandlimitedSignal}
\end{definition}
\begin{definition}[$S$-Time-Limitedness]
	The signal $\boldsymbol{x}\in \mathbb{C}^{N}$ is a $S$-time-limited signal if $\boldsymbol{x}\in \mathbb{C}^{N}$ contains $N-S$ consecutive zeros. That is, there exists $k$ such that $\boldsymbol{x}[k]=\cdots=\boldsymbol{x}[N+k-S-1]=0$.
	\label{def:timelimited}
\end{definition}
We first focus on the band-limited case, wherein our goal is to uniquely identify a band-limited signal $\boldsymbol{x}$, up to trivial ambiguities, from $\boldsymbol{B}$ under mild conditions stated in the following Proposition~\ref{theo:uniqueness}. 
\begin{proposition}
	Assume $\boldsymbol{x}\in \mathbb{C}^{N}$ to be a $B$-band-limited signal as in Definition \ref{def:bandlimitedSignal} for some $B\leq N/2$. Then, almost all signals are uniquely determined from $\boldsymbol{B}[k_{1},k_{2}]$, up to trivial ambiguities, from $m\geq 3B$ measurements. If, in addition, the signal's power spectrum is known and $N \geq 3$, then $m\geq2B$ suffice.
	\label{theo:uniqueness}
	\vspace{-0.5em}
\end{proposition}
The \textit{almost all} signals in Proposition \ref{theo:uniqueness} imply that the set of signals, which cannot be uniquely determined up to trivial ambiguities, is contained in the vanishing locus of a nonzero polynomial.	\vspace{-0.5em}
\begin{proof}	
	Assume that $\boldsymbol{y}\in \mathbb{C}^{N}$ is the Fourier transform of $\boldsymbol{x}\in \mathbb{C}^{N}$. Also, consider that $B = N/2$, $N$ is even, that $\boldsymbol{y}[n] \not= 0$ for $k = 0\dots,B-1$, and that $\boldsymbol{y}[n] = 0$ for $k = N/2,\dots, N-1$. If the signal's nonzero Fourier coefficients are not in the interval $0,\dots, N/2-1$, then we can cyclically re-index the signal without affecting the proof. If $N$ is odd, then one should replace $N/2$ by $\lfloor N/2 \rfloor$ everywhere in the sequel.
	
	Observe that the band-limited assumption on the signal forms the following structure on $\boldsymbol{B}[k_{1},k_{2}]$, where each row represents a fixed $k_{1}$ and varying $k_{2}$ for $k_{1}=0,\dots,N/2-1$\vspace{-0.5em}
	\begin{align}
		\label{eq:piramid}
		&\boldsymbol{y}[0]\lvert \boldsymbol{y}[0]\rvert^{2}, 0,\dots,0,0,\dots,0 \nonumber\\
		& \boldsymbol{y}^{2}[1] \overline{\boldsymbol{y}}[0], \lvert \boldsymbol{y}[1] \rvert^{2} \boldsymbol{y}[0], 0,\dots,0,0,\dots,0  \nonumber\\
		&\hspace{10em}\vdots\nonumber\\
		&\boldsymbol{y}^{2}[B-1] \overline{\boldsymbol{y}}[0],\boldsymbol{y}[B-1]\boldsymbol{y}[1]\boldsymbol{y}[B-2],\dots,\lvert \boldsymbol{y}[B-1] \rvert^{2} \boldsymbol{y}[0],0,\dots,0\hspace{5em} \nonumber\\
		&\hspace{10em}\vdots \nonumber\\
		&0,0,0,\dots,0,0,\dots,0.
	\end{align}
	
	\textbf{Step 0: }From the $0$-th row of \eqref{eq:piramid}, we see that $\boldsymbol{B}[0,0]= \boldsymbol{y}[0]\lvert \boldsymbol{y}[0]\rvert^{2}.$ Considering that the translation ambiguity is continuous, we can set $\boldsymbol{y}[0]$ to be real and, without loss of generality, it can be assumed that $\boldsymbol{y}[0]=1$ \cite{pinilla2021banraw}. Then, we obtain $\boldsymbol{B}[k,k] = \lvert \boldsymbol{y}[k] \rvert^{2}$ (main diagonal of BF), and $\boldsymbol{B}[k,0]=\boldsymbol{y}^{2}[k]$(first column of BF), $\forall k=1,\dots,B-1.$
	
	\textbf{Step 1: }From \textbf{Step 0} we have that $\boldsymbol{y}[0]$, $\{ \lvert \boldsymbol{y}[k] \rvert^{2}\}_{k=1}^{B-1}$, and $\{ \boldsymbol{y}^{2}[k]\}_{k=1}^{B-1}$ are known. There are up to $2^{B-1}$ vectors, modulo global phase that satisfy the constraints $\boldsymbol{B}[k,0] = \boldsymbol{y}^{2}[k]$ (because of the sign).
	
	\textbf{Step 2: }Fix one of the possible solutions for $\boldsymbol{y}[1]$ from \textbf{Step~1}. Moving to analyze the third row of \eqref{eq:piramid} we have that $\boldsymbol{y}[2]$ can be uniquely determined because $\boldsymbol{B}[2,1]=\boldsymbol{y}[2]\lvert \boldsymbol{y}[1]\rvert^{2}$. \vspace{0.5em}
	
	Despite a large number of possible solutions, we can prove that, at this step, there is only one vector (up to trivial ambiguities) out of the $2^{B-1}$ possibilities of \textbf{Step 2}, that is consistent with the equality constraints in the previous steps. We show this as follows: from \textbf{Step 1}, we have the knowledge of $\lvert \boldsymbol{y}[1]\rvert$, and $\boldsymbol{y}^{2}[1]$, which implies that a unique selection of $\boldsymbol{y}[1]$ in \textbf{Step 2} exists. Consequently, a unique selection of $\boldsymbol{y}[2]$ exists in this step because $\log(\boldsymbol{y}^{2}[1]) = 2\log(\boldsymbol{y}[1]) = 2(i\theta_{1} + \log(\lvert \boldsymbol{y}[1]\rvert))$, where $\boldsymbol{y}[1] = \lvert \boldsymbol{y}[1]\rvert e^{i\theta_{1}}$ implying that the phase $\theta_{1}$ is uniquely determined.
	
	\textbf{Step $B-1$: }Considering that from the $B-2$ previous steps, the entries $\boldsymbol{y}[0],\dots,\boldsymbol{y}[B-2]$ were uniquely determined (up to trivial ambiguities), then performing an analogous analysis as in \textbf{Step 2} uniquely determines $\boldsymbol{y}[B-1]$.
	
	Observe that at \textbf{Step 2}, the signal $\boldsymbol{y}$ is uniquely determined meaning that $m\geq 3B$ measurements are needed to solve this problem. If, in addition, the signal's power spectrum is known and $N \geq 3$, then $m\geq2B$ suffice. \vspace{-0.6em}
\end{proof}

A direct consequence of Proposition \ref{theo:uniqueness} is the following Corollary~\ref{coro:time}, which similarly states recovery of almost all time-limited signals as in Definition \ref{def:timelimited}, under mild conditions.
\begin{corollary}
	Assume $\boldsymbol{x}\in \mathbb{C}^{N}$ to be a $S$-time-limited signal as in Definition \ref{def:timelimited} for some $S\leq N/2$. Then, almost all signals are uniquely determined from their third-order cumulant (equivalent to BF in the time domain), up to trivial ambiguities, from $m\geq 3S$ measurements. If, in addition, we have access to the signal's power and $N \geq 3$, then $m\geq2S$ measurements suffice.
	\label{coro:time}
\end{corollary}
\begin{proof}
	The third-order cumulant, which is equivalent to BF in the time-domain, is 
	\begin{align}
		\boldsymbol{C}[n_{1},n_{2}] = \frac{1}{N}\sum_{n=0}^{N-1} \boldsymbol{x}[n]\overline{\boldsymbol{x}}[n-n_{1}]\boldsymbol{x}[n+n_{2}].
	\end{align}
	The proof follows, \textit{mutatis mutandis}, by performing an analogous construction procedure as in Proposition \ref{theo:uniqueness} over \eqref{eq:piramid}. 
\end{proof}

Evidently, Proposition \ref{theo:uniqueness} and Corollary \ref{coro:time} state that  all frequencies (time-shifts) of the BF (third-order cumulant) are \textit{not} needed to recover $\boldsymbol{x}$. Therefore, a recovery algorithm for this regime is desired. The proof also reveals that the first and the $(B-1)/(S-1)$-th rows of the BF and the third-order cumulant must be perfectly preserved in order to ensure uniqueness (up to ambiguities). 

Recent works have shown that minimizing the amplitude least-squares objective leads to better estimations under noisy scenarios \cite{zhang2016reshaped,8410803}. The objective to recover $\boldsymbol{y}$ is \vspace{-1em}
\begin{align}
	\min_{\boldsymbol{y}\in \mathbb{C}^{N}} h(\boldsymbol{y}) = \sum_{k_{1},k_{2}=0}^{N-1}\left \lvert \boldsymbol{B}[k_{1},k_{2}]- \boldsymbol{y}[k_{1}]\overline{\boldsymbol{y}}[k_{2}]\boldsymbol{y}[k_{1}-k_{2}] \right\rvert^{2},
	\label{eq:auxproblem}
\end{align}
We propose a trust region algorithm based on the Cauchy point (which lies on the gradient and minimizes the quadratic cost function by iteratively refining the trust region of the solution) to solve \eqref{eq:auxproblem} that we initialize via a spectral procedure.\vspace{-0.8em}

\section{Reconstruction Algorithm}
\label{sec:algorithm}\vspace{-0.5em}
The standard update rule in trust region methods takes the form $\boldsymbol{x}^{(t+1)}:=\boldsymbol{x}^{(t)} + \alpha^{(t)}\boldsymbol{b}^{(t)},$ where $\alpha^{(t)}$ is the step size at iteration $t$ and the vector $\boldsymbol{b}^{(t)}$ is chosen as
\begin{align}
	\argmin_{\boldsymbol{b}\in \mathbb{C}^{N}} &\hspace{0.5em}h(\boldsymbol{x}^{(t)}) + 2\mathcal{R}\left(\boldsymbol{b}^{H}\boldsymbol{d}^{(t)}\right), \hspace{1em} s.t \hspace{0.5em}\lVert \boldsymbol{b} \rVert_{2}\leq \Delta^{(t)},
	\label{eq:trust}
\end{align}
with $\mathcal{R}(\cdot)$ as the real part function, and $\boldsymbol{d}^{(t)}$ as the gradient of $h(\boldsymbol{z})$ with respect to $\overline{\boldsymbol{z}}$ at iteration $t$. The solution to \eqref{eq:trust} is $\boldsymbol{b}^{(t)} = - \frac{\Delta^{(t)}}{\lVert \boldsymbol{d}^{(t)} \rVert_{2}} \boldsymbol{d}^{(t)},$ according to \cite[Chapter 4]{nocedal2006numerical}.

To alleviate the memory requirements and computational complexity required for large $N$, we suggest a block stochastic gradient descent strategy. Then, the block stochastic version of $\boldsymbol{d}^{(t)}$, using the Wirtinger derivative (a partial first-order derivative of complex variables with respect to a real variable), is 
\begin{align}
	\label{eq:functionh}
	&\boldsymbol{d}_{\Gamma_{(t)}}[p]=-\sum_{k\in \Gamma_{(t)}} \boldsymbol{y}[k]\boldsymbol{y}[k-p]\overline{\boldsymbol{B}}[k,p] \nonumber\\
	&-\sum_{k\in \Gamma_{(t)}} \boldsymbol{y}[k]\overline{\boldsymbol{y}}[p-k]\boldsymbol{B}[p,k] - \sum_{k\in \Gamma_{(t)}} \boldsymbol{B}[k,k-p]\boldsymbol{y}[k-p]\overline{\boldsymbol{y}}[k] \nonumber\\
	&+\sum_{k\in \Gamma_{(t)}} \boldsymbol{y}[p]\lvert \boldsymbol{y}[k] \rvert^{2}\left( \lvert \boldsymbol{y}[p-k] \rvert^{2} + \lvert \boldsymbol{y}[k-p] \rvert^{2}\right) \nonumber\\
	&+ \sum_{k\in \Gamma_{(t)}} \boldsymbol{y}[p]\lvert \boldsymbol{y}[k+p] \rvert^{2}\lvert \boldsymbol{y}[k] \rvert^{2}.
\end{align}
The set $\Gamma_{(t)}$ is chosen uniformly and independently at random at each iteration $t$ from subsets of $\{1,\cdots,N \}^{2}$ with cardinality $Q$. In this way, the gradient is uniformly sampled using a minibatch of data of size $Q$ for each step update, such that, in expectation, it is the true gradient \cite[page 130]{spall2005introduction}. 

For initialization, we consider the first column of the BF. 
From \eqref{eq:BF}, define $\boldsymbol{s}\in \mathbb{C}^{N}$ as $\boldsymbol{s}^{2}[n] = \frac{1}{\overline{\boldsymbol{y}}[0]}\boldsymbol{B}[n,0] = \boldsymbol{y}[n]^{2},$ for $n=0,\dots,N-1$. Following Proposition \ref{theo:uniqueness}, the value of $\boldsymbol{y}[0]$ is determined from $\boldsymbol{B}[0,0]$ and it is assumed to be different than zero for almost all signals. Recall that $\boldsymbol{y}$ is the Fourier transform of $\boldsymbol{x}$, which is approximated as $\boldsymbol{s}[n] = \pm \sqrt{\frac{1}{\overline{\boldsymbol{y}}[0]}\boldsymbol{B}[n,0]} \equiv \boldsymbol{y}[n]$. However, the above inversion yields $2^{N}$ possible solutions. 

Fortunately, the Fourier transform of the underlying signal enjoys several strong constraints. First, the spectrum is always square-integrable following the finite energy of the pulse. Also, the spectrum is always positive. Then, following the Paley-Wiener theorem, the inverse Fourier transform $\boldsymbol{s}$ is a $C^{\infty}$ function, which is continuous and has infinitely many continuous derivatives at every (temporal) point \cite{strichartz2003guide}. Note that the $C^{\infty}$ property over the spectrum is not related to its dimension. Rather, it implies that for discrete differentiation over the signal (via the finite differences method because the signal is a vector), the shape of the resultant signal has to be continuous because it is $C^{\infty}$. Then, exploiting the continuity of $\boldsymbol{y}$ we can eliminate all ambiguities. 

We summarize this process below in Algorithm \ref{alg:initialization}, which starts at some point $n_{0}\in \{ 0,N-1\}$ and chooses the root $\boldsymbol{s}[n_{0}] = \sqrt{\frac{1}{\overline{\boldsymbol{y}}[0]}\boldsymbol{B}[n,0]}.$ Then, we find the two possible roots for the next temporal point, $n_{p}+1$ for $p=0,\dots,N-2$. If $\lvert \boldsymbol{s}[n_{p}] \rvert$ and $\lvert \boldsymbol{s}[n_{p}+1] \rvert$ are both nonzero, one of the roots $\lvert \boldsymbol{s}_{\pm}[n_{p}+1] \rvert$ will necessarily be considerably closer to $\lvert \boldsymbol{s}[n_{p}] \rvert$ than its additive inverse. So, we eliminate the more distant root from Line 9 to Line 12, and it is then on to the next temporal point. While this procedure may be continued \textit{ad infinitum}, Paley-Wiener theorem implies we  go only as far as the second derivative by building $\Delta_{0},\Delta_{1}$, and $\Delta_{2}$ in Lines 5, 6, and 7, respectively. Once the main loop of Algorithm \ref{alg:initialization} is finished, we return the inverse Fourier transform of $\boldsymbol{s}$ as the estimate of the signal $\boldsymbol{x}$. 
\begin{algorithm}[H]
	\caption{Initialization Procedure}
	\label{alg:initialization}
	\small
	\begin{algorithmic}[1]
		\State{\textbf{Input: }$\left\lbrace\boldsymbol{B}[k_{1},k_{2}]:k_{1},k_{2}=0,\cdots,N-1 \right\rbrace$. Choose $n_{0}\in \{0,\dots,N-1 \}$, and $\gamma_{0}=0.09,  \gamma_{1}=0.425,\gamma_{2}=1.0$}
		\State{\textbf{Output:} $\boldsymbol{x}^{(0)}$ (estimation of $\boldsymbol{x}$).}
		\State{\textbf{Notation: } $\boldsymbol{s}_{\pm}[n]$ and $\epsilon_{\pm}$ stand for evaluating both roots. $\mathcal{F}^{-1}$ denotes inverse Fourier transform.}
		\For{$p=0$ to $N-2$}
		\State{$\Delta_{0,\pm} \gets \boldsymbol{s}_{\pm}[n_{p} + 1] - \boldsymbol{s}[n_{p}]$}
		\State{$\Delta_{1,\pm} \gets (\boldsymbol{s}_{\pm}[n_{p} + 1] - \boldsymbol{s}[n_{p}]) - (\boldsymbol{s}_{\pm}[n_{p}] - \boldsymbol{s}[n_{p}-1])$}
		\State{$\Delta_{2,\pm} \gets (\boldsymbol{s}_{\pm}[n_{p} + 1] - \boldsymbol{s}[n_{p}]) - (\boldsymbol{s}_{\pm}[n_{p}] - \boldsymbol{s}[n_{p}-1]) - (\boldsymbol{s}_{\pm}[n_{p}-1] - \boldsymbol{s}[n_{p}-2])$}
		\State{$\epsilon_{\pm} \gets \gamma_{0} \lvert \Delta_{0,\pm}\rvert^{2} + \gamma_{1}\lvert\Delta_{1,\pm} \rvert^{2} + \gamma_{2}\lvert \Delta_{2,\pm}\rvert^{2}$}
		\If{$\epsilon_{+} > \epsilon_{-}$}
		\State{$\boldsymbol{s}[n_{p}+1] = \sqrt{\frac{1}{\overline{\boldsymbol{y}}[0]}\boldsymbol{B}[n,0]}$}
		\Else
		\State{$\boldsymbol{s}[n_{p}+1] = -\sqrt{\frac{1}{\overline{\boldsymbol{y}}[0]}\boldsymbol{B}[n,0]}$}
		\EndIf
		\State{$n_{p} \gets n_{p} +  1$}
		\EndFor
		\State{\textbf{return: }$\boldsymbol{x}^{(0)}:= \mathcal{F}^{-1}(\boldsymbol{s})$.}
	\end{algorithmic}
\end{algorithm}

Note that the constants $\gamma_{0}, \gamma_{1}$, and $\gamma_{2}$ are chosen using a cross-validation strategy such that it provides the closest estimation of the true $\boldsymbol{y}$. Algorithm~\ref{alg:smothing} summarizes the entire BF inversion procedure.

\begin{algorithm}[H]
	\caption{Signal recovery from the bispectrum function}
	\label{alg:smothing}
	\small
	\begin{algorithmic}[1]
		\State{\textbf{Input: }Data $\left\lbrace\boldsymbol{B}[k_{1},k_{2}]:k_{1},k_{2}=0,\cdots,N-1 \right\rbrace$. Choose $\alpha^{(0)}\in(0,1)$, $Q\in \{1,\cdots,N^{2}\}$, $\epsilon>0$, $\Delta^{(0)}\in (0, \frac{1}{Q}\lVert \boldsymbol{x} \rVert_{2})$, and $\gamma,\gamma_{1},\delta_{1},\delta_{2}\in (0,1)$.}
		\State {Initial point $\boldsymbol{x}^{(0)} \leftarrow$ Algorithm 2$\left(\boldsymbol{B}[k_{1},k_{2}]\right)$.}
		\While{$\left\lVert\boldsymbol{b}_{\Gamma_{(t)}}\right\rVert_{2}\geq \epsilon$}
		\Statex{Choose $\Gamma_{(t)}$ uniformly at random from the subsets of $\{1,\cdots,N \}^{2}$ with cardinality $Q$ per iteration $t\geq 0$.}
		\State{Compute $\boldsymbol{d}_{\Gamma_{(t)}}$ and $\boldsymbol{b}^{(t)} = - \frac{\Delta^{(t)}}{\lVert \boldsymbol{d}_{\Gamma_{(t)}} \rVert_{2}} \boldsymbol{d}_{\Gamma_{(t)}}$ and  $\rho = 1$}
		\While{\small{$h( \boldsymbol{x}^{(t)} + \rho \boldsymbol{b}^{(t)} ) > h( \boldsymbol{x}^{(t)}) + \delta_{1}\rho \mathcal{R}(\boldsymbol{d}_{\Gamma_{(t)}}^{H} \boldsymbol{b}^{(t)})$}}
		\State $\rho = \delta_{2} \rho $ 
		\EndWhile
		\State{Set $\alpha^{(t)} = \rho$, and $\displaystyle\boldsymbol{x}^{(t+1)}=\boldsymbol{x}^{(t)} + \alpha^{(t)} \boldsymbol{b}_{\Gamma_{(t)}}$} 
		\If{$\displaystyle\left\lVert\boldsymbol{b}_{\Gamma_{(t)}}\right\rVert_{2}\geq \gamma \Delta^{(t)}$}
		\State{$\Delta^{(t+1)}=\Delta^{(t)}$.}
		\Else
		\State{$\Delta^{(t+1)} = \gamma_{1}\Delta^{(t)}$.}
		\EndIf
		\EndWhile
		\State{\textbf{return: } $\boldsymbol{x}^{(T)}$.}
	\end{algorithmic}
\end{algorithm}\vspace{-2em}

\vspace{-1em}

\section{Numerical Results}
\label{sec:numerical}
We evaluated the performance of Algorithm \ref{alg:smothing} through numerical experiments, wherein we used the following parameters: 
$\gamma_{1}=0.1$, $\gamma=0.1$, $\alpha = 0.6$, and $\epsilon=1\times 10^{-4}$. The number of indices that are chosen uniformly at random is fixed as $Q=N$. We built a set of $\left\lceil \frac{N-1}{2} \right\rceil$-band-limited signals that conform to a Gaussian power spectrum. 
Specifically, each signal ($N=128$ grid points) is produced via the Fourier transform of a complex vector with a Gaussian-shaped amplitude. 
Next, we multiply the obtained power spectrum by a uniformly distributed random phase. All simulations were implemented in Matlab R2021a on an Intel Core i5 2.5Ghz CPU with 16 GB RAM. 

The performance of Algorithm \ref{alg:smothing} is evaluated for complete and incomplete noisy BFs, where the signal-to-noise-ratio (SNR) is defined as SNR$ = 10\log_{10}(\lVert \mathbf{B} \rVert^{2}_{\text{F}}/\lVert \boldsymbol{\sigma} \rVert^{2}_{\text{2}})$, with $\boldsymbol{\sigma}$ as the variance of the noise. The BF is incomplete when few vertical frequencies, i.e. $k_{1}$ dimension, are uniformly removed. We measure the relative error between the true signal $\boldsymbol{x}{x}$ and any $\boldsymbol{w}\in \mathbb{C}^{N}$ as $\text{dist}(\boldsymbol{x},\boldsymbol{w}):= \frac{\left\lVert \boldsymbol{B}-\boldsymbol{W} \right\rVert_{\text{F}}}{\left\lVert \boldsymbol{B} \right\rVert_{\text{F}}},$ where $\boldsymbol{W}$ is the BF of $\boldsymbol{w}$. 

\begin{figure}[t]
	\centering
	\includegraphics[width=0.8\linewidth]{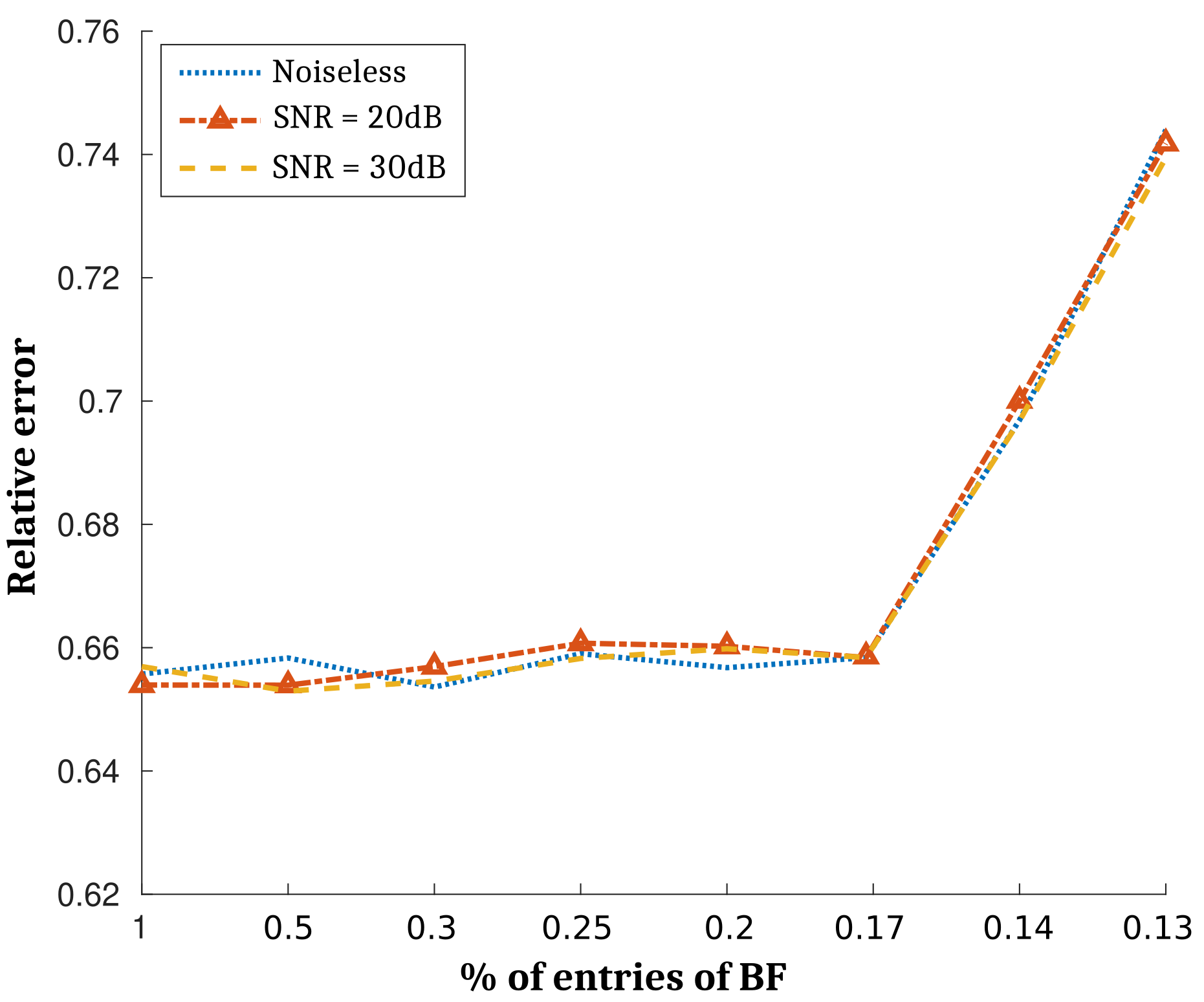}
	\caption{Performance of $\boldsymbol{x}^{(0)}$ obtained by Algorithm \ref{alg:initialization} at different SNR levels, for different percentages of (uniformly) removed frequencies in the $k_{1}$ dimension. The relative error was averaged over 100 trials.} \vspace{-1em}
	\label{fig:init}
\end{figure}

\begin{figure}[t]
	\centering
	\includegraphics[width=0.8\linewidth]{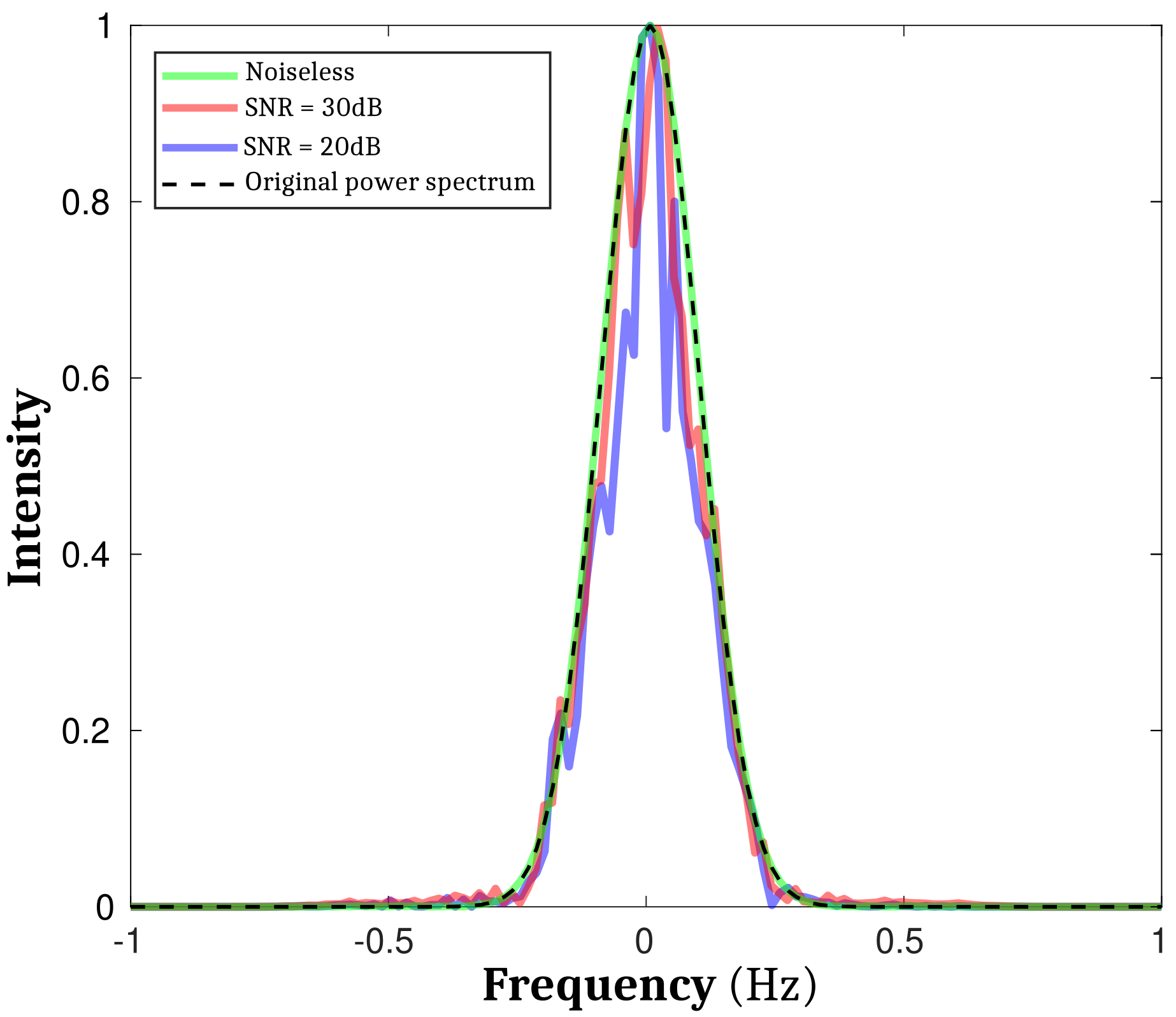}
	\caption{Example of an estimated power spectrum of the underlying signal $\boldsymbol{x}$ using Algorithm \ref{alg:initialization}. The noisy cases have SNR$=20$ and $30$ dB.}\vspace{-1em}
	\label{fig:power}
\end{figure}

We examine the performance of vector $\boldsymbol{x}^{(0)}$ obtained by Algorithm \ref{alg:initialization} to approximate the pulse $\boldsymbol{x}$ with complete as well as sparse samples of the BF. We numerically determine the relative error averaged over 100 trials for both noiseless and noisy scenarios for SNR $=20$ and $30$ dB. Fig. \ref{fig:init} shows the effectiveness of Algorithm \ref{alg:initialization} to estimate the underlying pulse. To complement the numerical results in Fig. \ref{fig:init}, we present the power spectrum estimate of $\boldsymbol{x}$. Fig.~\ref{fig:power} shows an approximated version of $\lvert \boldsymbol{y} \rvert^{2} $ to visually illustrate the ability of Algorithm \ref{alg:initialization} to correctly estimate $\boldsymbol{y}$ 
while uniquely determining the roots of $\boldsymbol{s}[n]$ to approximate the underlying signal $\boldsymbol{x}$ even in the presence of Gaussian noise.

\begin{figure}[t]
	\centering
	\includegraphics[width=0.8\linewidth]{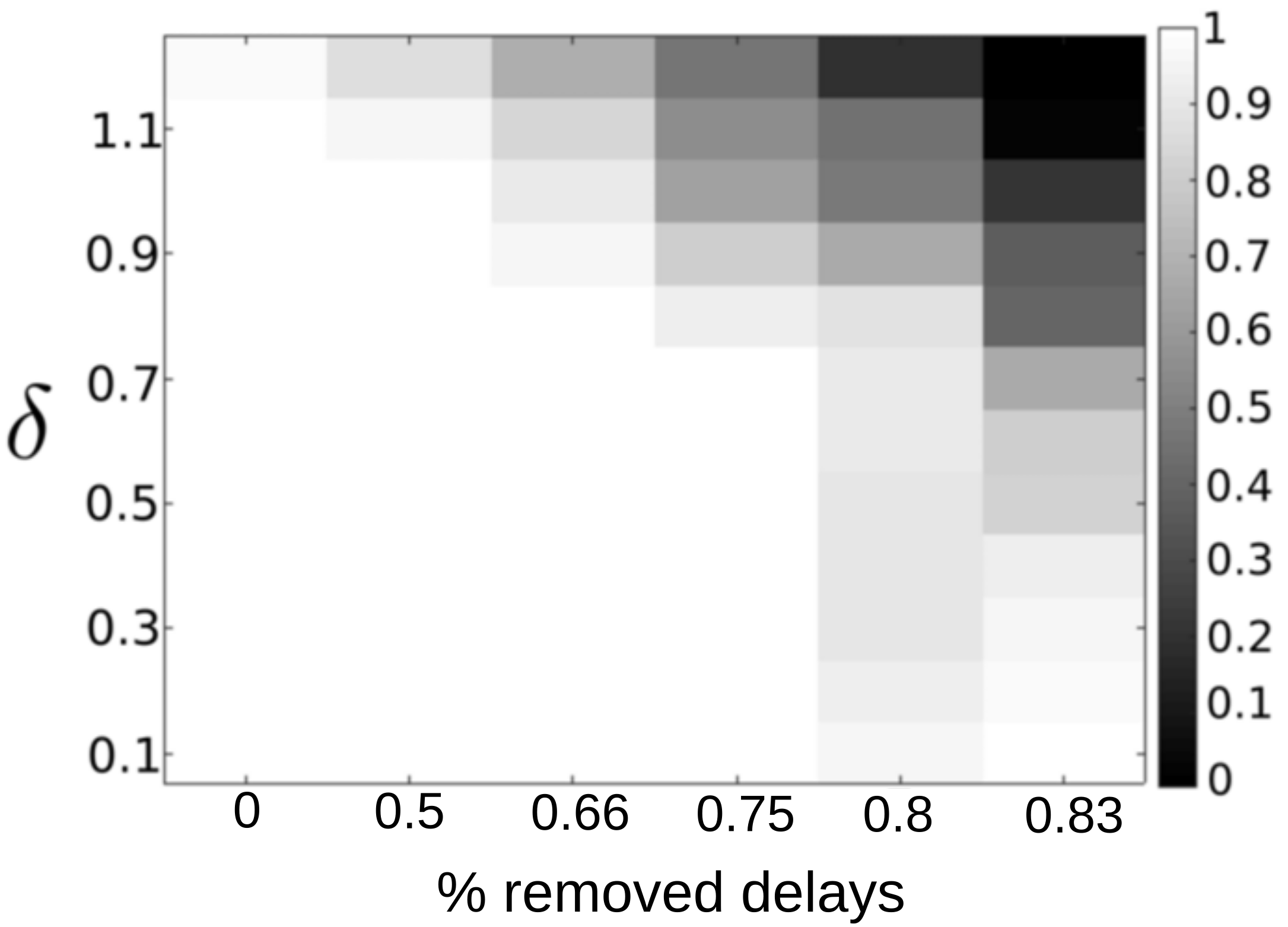}
	\caption{Empirical success rate of Algorithm \ref{alg:smothing} as a function of \% (uniformly) removed delays and $\delta$ in the absence of noise.}\vspace{-0.5em}
	\label{fig:orthoini}
	\vspace{-1em}
\end{figure}

Finally, we evaluated the success rate of Algorithm \ref{alg:smothing} in Fig. \ref{fig:orthoini}. Here, Algorithm \ref{alg:smothing} is initialized at $\boldsymbol{x}^{(0)} = \boldsymbol{x}+\delta\zeta$, where $\delta$ is a fixed constant and $\zeta$ takes values on $\{-1,1\}$ with equal probability, while a percentage of the delays are set to zero. A trial is declared successful when the returned estimate attains a relative error that is smaller than $10^{-6}$. We numerically determine the empirical success rate among 100 trials. Fig. \ref{fig:orthoini} shows that Algorithm \ref{alg:smothing} is able to estimate the pulse when the BF is incomplete.

\section{Summary}
\label{sec:conclusion}
We analytically demonstrated that band-limited signals can be estimated (up to trivial ambiguities) from its BF through our trust region gradient method. We evaluated our proposed technique for complete/incomplete noisy and noiseless scenarios. 
In the case of incomplete data, we found that although Proposition \ref{theo:uniqueness} suggests that the full BF is not required to guarantee uniqueness, there is potential to better estimate the pulses from incomplete data. Numerical experiments showed reasonably accurate performance for recovering both the magnitude and phase of the signal even from noisy incomplete data.

\section{Acknowledgment}
K. V. M. acknowledges partial support from the National Academies of Sciences, Engineering, and Medicine via the Army Research Laboratory Harry Diamond Distinguished Fellowship. Research was sponsored by the Army Research Laboratory and was accomplished under Cooperative Agreement Number W911NF-21-2-0288. The views and conclusions contained in this document are those of the authors and should not be interpreted as representing the official policies, either expressed or implied, of the Army Research Laboratory or the U.S. Government. The U.S. Government is authorized to reproduce and distribute reprints for Government purposes notwithstanding any copyright notation herein. S. P. acknowledges support from the EMET Research Institute in Colombia.

\bibliographystyle{IEEEtran}
\bibliography{report}

\end{document}